\documentclass{ifacconf}

\usepackage{amsmath}
\usepackage{amssymb}
\usepackage{graphicx}
\DeclareGraphicsExtensions{.pdf,.png,.jpg,.eps}

\newtheorem{theorem}{Theorem}
\newtheorem{proof}{Proof}

\newtheorem{remark}{Remark}

\usepackage{graphicx}      
\usepackage{natbib}        
\begin{document}
\begin{frontmatter}

\title{Feedback Control of Dynamical Systems with Given Restrictions on Output Signal\thanksref{footnoteinfo}} 

\thanks[footnoteinfo]{The results of Section 3 were developed under support of RSF (grant 18-79-10104) in IPME RAS. The other researches were partially supported by grants of Russian Foundation for Basic Research No. 19-08-00246 and Government of Russian Federation, Grant 074-U01.}

\author[First,Second]{Igor Furtat} 

\address[First]{Institute for Problems of Mechanical Engineering Russian Academy of Sciences, 61 Bolshoy ave V.O., St.-Petersburg, 199178, Russia, (e-mail: cainenash@mail.ru).}
\address[Second]{ITMO University, 49 Kronverkskiy ave, Saint Petersburg, 197101, Russia}

\begin{abstract}                
A novel method for control of dynamical systems, proposed in the paper, ensures an output signal belonging to the given set at any time. 
The method is based on a special change of coordinates such that the initial problem with given restrictions on an output variable can be performed as the problem of the input-to-state stability analysis of a new extended system without restrictions.
The new control laws for linear plants, systems with sector nonlinearity and systems with an arbitrary relative degree are proposed. 
Examples of change of coordinates are given, and they are utilized to design the control algorithms. The simulations confirm theoretical results and illustrate the effectiveness of the proposed method in the presence of parametric uncertainty and external disturbances.
\end{abstract}


\end{frontmatter}

\section{Introduction}

The control with guaranteeing the desired quality of transients in an output signal is an important problem of the theory and practice of automatic control. 
If plant parameters are known, there are numbers of classical methods are used: control methods with the placement of eigenvalues, control with the frequency response analysis, optimal control methods, etc., see, e.g. \cite{Kuo75,Golnaraghi17}. 
The problem of improving the upper bound of deviation of the output signal in linear systems with nonzero initial conditions is still relevant (\cite{Whidborne11,Polyak15}).
 
The methods of adaptive and robust control are effective under parametric uncertainty and disturbances, see, for example, \cite{Ioannou95,Fradkov00,Tao03}. 
The transient quality is specified by a reference model. 
However, the methods \cite{Ioannou95,Fradkov00,Tao03} do not guarantee a given deviation of the output signal from the reference signal in transient mode. 
If the plant initial conditions are unknown, then at the initial time these deviations can be sufficiently large. 
The methods \cite{Ioannou95,Fradkov00,Tao03} guarantee only the prespecified deviation of the output signal from the reference signal in the steady state. 
However, the estimation of prespecified deviation can be sufficiently rough.

The method \cite{Polyak14} ensures that output signals belong to the smallest ellipsoid in transition and steady state. 
However, this ellipsoid remains the same at any time, therefore, the method can give rough quality in transition and steady state.

The paper \cite{Miller91} proposes the adaptive control method which ensures belonging of output signal to given sets. 
These sets may be different for transient and steady state modes. 
The sets are performed by a sequence of rectangles. 
The height of each rectangle corresponds to the desired maximum deviation of the output variable from the equilibrium position. 
The length of the rectangle corresponds to the desired time when the output variable belongs to the corresponding rectangle. 
However, the rectangular areas in \cite{Miller91} are rather rough and the algorithm is applicable only for plants with scalar input and output signals.

Differently from \cite{Miller91}, in the paper \cite{Bechlioulis08} a control method with the guarantee of belonging the output signal to a given set for plants with vector input and vector output is proposed. However, the implementation of this method requires knowledge of the sign and knowledge of the set of initial conditions. 
Moreover, obtained upper and lower bounds for transients are rather rough because these bounds are determined by the same function with different signs. 
Additionally, the upper and lower bounds asymptotically converge to some constants.

In the present paper, we propose a new control method with providing an output signal to a given set. 
Differently from \cite{Bechlioulis08}, the given set can be described by functions that independent on the sign of plant initial conditions. 
Only knowledge of the set of initial values is required. 
Also, unlike \cite{Miller91,Bechlioulis08}, the configuration of the given set can be described by arbitrary continuously differentiable functions for which asymptotic convergence is not required. 
As a result, the obtained method significantly expands the class of tasks compared with \cite{Miller91,Bechlioulis08}.

The paper is organized as follows. In Section \ref{Sec2} the control problem is formulated. 
Section \ref{Sec3} describes the main result, where a special change of coordinate is proposed. As a result, the initial problem with restrictions can be performed as the problem of the input-to-state stability analysis of a new extended dynamical system without restrictions. 
Also in Section \ref{Sec3} examples of coordinate change are given. 
Section \ref{Sec4} proposes a state feedback control algorithm for linear plants with known parameters and unknown external bounded disturbances. 
Section \ref{Sec5} considers a synthesis of the output feedback control law for systems with sector nonlinearity. The proposed control law does not depend on the plant parameters. 
In Section \ref{Sec6} the new output feedback control law is designed for systems with an arbitrary relative degree. 
Also, in Sections \ref{Sec4}-\ref{Sec6} the simulations illustrate confirmation of theoretical results and show the effectiveness of the proposed method in the presence of parametric uncertainty and external disturbances.

\textit{Notations}. Throughout the paper the superscript $\rm T$ stands for matrix transposition; 
$\mathbb R^{n}$ denotes the $n$ dimensional Euclidean space with vector norm $|\cdot|$; 
$\mathbb R^{n \times m}$ is the set of all $n \times m$ real matrices; 
$I$ is the identity matrix of corresponding order;
$A^*$ is the adjugate of the matrix $A$.

\section{Problem formulation} \label{Sec2}

Consider a dynamical system in the form
\begin{equation}
\label{eq2_1}
\begin{array}{l} 
\dot{x}=F(x,u,t),
\\
y=h(x),
\end{array}
\end{equation} 
where $t \geq 0$, $x \in \mathbb R^n$ is the state vector, $u \in \mathbb R^m$ is the control signal, $y=col\{y_1,...,y_v\}$ is the output signal. The vector function $F$ is defined for all $x$, $u$, $t$ and it is a piecewise continuous and bounded function in $t$. The function $h(x)$ is continuously differentiable w.r.t. $x$. Plant \eqref{eq2_1} is controllable and observable for all $x \in \mathbb R^n$.

Our objective is to design a control law that ensures the input-to-state stability (ISS) of the closed-loop system and the signal $y(t)$ belongs to the following set
\begin{equation}
\label{eq2_20}
\begin{array}{l} 
\mathcal{Y}=\left\{y \in \mathbb R^v:~ \underline{g}_i(t) < y_i(t)  < \overline{g}_i(t),~i=1,...,v\right\}
\end{array}
\end{equation} 
for all $t \geq 0$. Here $\underline {g}_i(t)$ and $\overline{g}_i(t)$ are bounded functions with their first time derivatives. These functions are chosen by the designer. For example, in control of multi-machine power systems \cite{Pavlov01}, it is required to ensure the conditions: $\underline{w} < w(t) < \overline{w}$ and $\underline {V}<V(t)<\overline {V}$ for all $t \geq 0$, where $w(t)$ is the frequency and  $V(t)$ is the output voltage.

Differently from \cite{Bechlioulis08}, goal \eqref{eq2_20} is independed on the sign of plant initial conditions. Also, unlike \cite{Miller91,Bechlioulis08}, the set $\mathcal{Y}$ in \eqref{eq2_20} can be described by arbitrary continuously differentiable functions for which asymptotic convergence is not required.

\section{Main result}
\label{Sec3}

Let us consider a change of the output variable $y(t)$ in the form
\begin{equation}
\label{eq2_2}
\begin{array}{l} 
y(t)=\Phi(\varepsilon(t),t),
\end{array}
\end{equation}
where $\varepsilon(t) \in \mathbb R^v$ is the continuously differentiable vector function w.r.t. $t$, 
the function 
$\Phi(\varepsilon,t)=col\{\Phi_1(\varepsilon,t),...,\\ \Phi_v(\varepsilon,t)\}$ satisfies the following conditions:

\begin{enumerate}
\item [(a)] $\underline{g}_i(t)<\Phi_i(\varepsilon,t)<\overline{g}_i(t)$, $i=1,...,v$ for all $t \geq 0$ and $\varepsilon \in \mathbb R^v$;
\item [(b)]
there exists the inverse function $\varepsilon=\Phi^{-1}(y,t)$ for all $y \in \mathcal{Y}$ and $t \geq 0$;
\item [(c)] the function $\Phi(\varepsilon,t)$ is continuously differentiable in 
$\varepsilon$ and $t$ as well as $\det\left(\frac{\partial \Phi(\varepsilon,t)}{\partial \varepsilon}\right) \neq 0$ for all $t \geq 0$ and $\varepsilon \in \mathbb R^v$;
\item [(d)] the function $\frac{\partial \Phi(\varepsilon,t)}{\partial t}$ is bounded on $t \geq 0$ for all $\varepsilon \in \mathbb R^v$.
\end{enumerate}

Consider several examples of the function $\Phi(\varepsilon,t)$.

\textit{Example 1.}
Let $\Phi (\varepsilon,t)=g(t)S(\varepsilon)$, where the function $S(\varepsilon) \in \mathbb R$ defines a coordinate change and the function $g(t) \in \mathbb R$ describes the boundary of a given restrictions. 
Additionally, $g(t) \neq 0$ and $\dot{g}(t)$ are bounded functions, 
$S(\varepsilon)=\frac{\varepsilon}{|\varepsilon|+1}+r$, $r \in \mathbb R$. 
See example of the function $S(\varepsilon)$ in Fig.~\ref{Fig000} for $r=1$. 
Since  $y(t)=g(t)S(\varepsilon)$ from \eqref{eq2_2}, then $(r-1)g(t)<y(t)<(r+1)g(t)$ for $g(t)>0$ and $(r+1)g(t)<y(t)<(r-1)g(t)$ for $g(t)<0$. 
The inverse function takes the form $\varepsilon=\frac{y-rg}{g-(y-rg) sign(\varepsilon)}$.

\textit{Example 2.}
In Example 1 introduce the function $S(\varepsilon)$ in the form $S(\varepsilon)=\frac{\overline{r}e^{\varepsilon}+\underline{r}}{e^{\varepsilon}+1}$, where 
$0<\underline{r}<\overline{r}$. See example of the function $S(\varepsilon)$ in Fig.~\ref{Fig000}for $\overline{r}=1.5$ and $\underline{r}=0.5$.
Then the inverse function $\varepsilon=\ln\frac{\underline{r}gy}{y-\overline{r}g}$ is valid for $\underline{r}g(t)<y(t)<\overline{r}g(t)$ and $g(t)>0$ or for $\overline{r}g(t)<y(t)<\underline{r}g(t)$ and $g(t)<0$.

In examples 1 and 2 the upper and lower boundaries of the given restrictions depend on the function $g(t)$. The following two examples contain the change of variable with independent functions of the given restrictions.

\begin{figure}[h!]
\center{\includegraphics[width=1\linewidth]{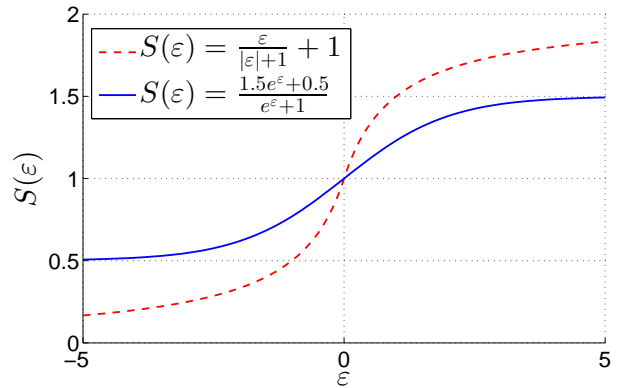}}
\caption{The plots of the functions $S(\varepsilon)=\frac{\varepsilon}{1+|\varepsilon|}+1$ and $S(\varepsilon)=\frac{1.5e^{\varepsilon}+0.5}{e^{\varepsilon}+1}$.}
\label{Fig000}
\end{figure}

\textit{Example 3.}
Let $\Phi(\varepsilon,t)=\frac{\overline{g}(t)e^{\varepsilon}+\underline{g}(t)}{e^{\varepsilon}+1}$, 
where $\Phi(\varepsilon,t) \in \mathbb R$, $\varepsilon \in \mathbb R$, 
the functions $\overline{g}(t)$, $\underline{g}(t)$, $\dot{\overline{g}}(t)$ and $\dot{\underline{g}}(t)$ are bounded for all $t$ and $\underline{g}(t)<\overline{g}(t)$.
Taking into account \eqref{eq2_2}, the inverse function $\varepsilon=\ln\frac{\underline{g}-y}{y-\overline{g}}$ is performed for $\underline{g}(t)<y(t)<\overline{g}(t)$ for all $t$

\textit{Example 4.}
Let $\Phi(\varepsilon,t)$ be presented in the form 
$\Phi(\varepsilon,t)=\begin{cases}
   \overline{g}(t)+0.5(\underline{g}(t)-\overline{g}(t))e^{-\varepsilon},~ \varepsilon \geq 0,\\
   \underline{g}(t)+0.5(\overline{g}(t)-\underline{g}(t))e^{\varepsilon},~ \varepsilon < 0,
\end{cases}$
where the functions $\overline{g}(t)$ and $\underline{g}(t)$ are the same as in Example 3.
Taking into account \eqref{eq2_2}, the inverse function takes the form 
$\varepsilon=\begin{cases}
   \ln\frac{\overline{g}-\underline{g}}{2(\overline{g}-y)},~ 0 \leq y < \overline{g},\\
   \ln\frac{2(y-\underline{g})}{\overline{g}-\underline{g}},~ \underline{g} <y < 0. 
\end{cases}$

Now we define the dynamics of the variable $\varepsilon(t)$ for the ISS analysis of the closed-loop system. Take the derivative of \eqref{eq2_2} w.r.t. $t$ and rewrite result as
$\dot{y}=\frac{\partial\Phi(\varepsilon,t)}{\partial \varepsilon}
\dot{\varepsilon}+\frac{\partial\Phi(\varepsilon,t)}{\partial t}$.
It follows from condition (c) that
$det\left(\frac{\partial\Phi(\varepsilon,t)}{\partial \varepsilon}\right) \neq 0$. Taking into account \eqref{eq2_1}, rewrite the dynamics of $\varepsilon(t)$ in the form
\begin{equation}
\label{eq2_3}
\begin{array}{l} 
\dot{\varepsilon}=\left(\frac{\partial\Phi(\varepsilon,t)}{\partial \varepsilon}\right)^{-1}
\left(\dot{y}-\frac{\partial\Phi(\varepsilon,t)}{\partial t}
\right).
\end{array}
\end{equation}

\begin{theorem}
\label{Th1}
Let conditions (a)-(d) hold for \eqref{eq2_2}. If there exists the control law $u$
such that the solutions of \eqref{eq2_1} and \eqref{eq2_3} are bounded, then $y(t) \in \mathcal Y_\alpha \subset \mathcal Y$.
If the solutions of \eqref{eq2_3} are unbounded, then $y(t) \in \mathcal Y_\beta \subseteq \mathcal Y$.
\end{theorem}

\begin{proof}
Let the control law $u$ be chosen such that the solutions of \eqref{eq2_3} are bounded. Then $|\varepsilon(t)|<N$ for all $t$, where $N>0$. 
According to \eqref{eq2_2},
$y \in \mathcal{Y}_\alpha=\big\{y \in \mathbb R^v: \underline M_i(t) \leq y_i(t) \leq \overline M_i(t),~i=1,...,v\big\}$ for all $t$,
where $\underline M_i(t)=\inf\limits_{|\varepsilon| \leq N}\{\Phi_i(\varepsilon,t)\}$ and $\overline M_i(t)=\sup\limits_{|\varepsilon| \leq N}\{\Phi_i(\varepsilon,t)\}$.
Since \eqref{eq2_2} is a bijective function, $\overline M_i(t)<\overline g_i(t)$ and $\underline M_i(t)>\underline g_i(t)$ for all $t$.

If the control law does not provide the boundedness of the solution of \eqref{eq2_3}, then
$y \in \mathcal{Y}_\beta=\big\{y \in \mathbb R^v: \underline S_i(t) < y_i(t)  < 
\\ \overline S_i(t),~i=1,...,v\big\}$,
where $\underline S_i(t)=\inf\limits_{\varepsilon \in \mathbb R^v}\{\Phi_i(\varepsilon,t)\}$ and $\overline S_i(t)=\sup\limits_{\varepsilon \in \mathbb R^v}\{\Phi_i(\varepsilon,t)\}$ for all $ t $.
Since \eqref{eq2_2} is a bijective function, $\overline S_i(t) \leq \overline g_i(t)$ and $\underline S_i(t) \geq \underline g_i(t)$ for all $ t $ .
Theorem \ref{Th1} is proved.
\end{proof}

In the next sections we will demonstrate the proposed method for some plants.

\section{State feedback control for linear plants under disturbances}
\label{Sec4}

Let the plant be described by the following linear differential equation
\begin{equation}
\label{eq2_4}
\begin{array}{l} 
\dot{x}=Ax+Bu+Df,
\\
y=Lx.
\end{array}
\end{equation}
The signals $x \in \mathbb R^n$, $u \in \mathbb R$, and $y \in \mathbb R$ are measured, $f \in \mathbb R^l$ is the unknown bounded disturbance, the matrices $A \in \mathbb R^{n \times n}$, $B \in \mathbb R^{n}$ and $L \in \mathbb R^{1 \times n}$ are known, the matrix $D$ is unknown. The pair $(A,B)$ is controllable and the pair $(L,A)$ is observable.

We formulate a result that contains the "simplest" $~$ control law in the sense of the "convenience" $~$ stability analysis of the closed-loop system.

\begin{theorem}
\label{Th02}
Let conditions (a)-(d) hold for transformation \eqref{eq2_2}, 
$\frac{\partial\Phi(\varepsilon,t)}{\partial \varepsilon}>0$ for all $\varepsilon$ and $t$, and there exists the vector $T \in \mathbb R^n$ such that the matrix $(I-(LB)^{-1}BL)A-TL$ is Hurwitz. Given $\alpha>0$ and $K>0$ there exists $\beta>0$ such that the linear matrix inequality (LMI)
\begin{equation}
\label{LMI_Th02}
\begin{array}{l}
\begin{bmatrix} \alpha-K & 0.5\\0.5 & -\beta \end{bmatrix} \leq 0
\end{array}
\end{equation}
holds. Then the control law
\begin{equation}
\label{eq2_6}
\begin{array}{l} 
u=-(LB)^{-1}\left[LAx+ K \varepsilon\right]
\end{array}
\end{equation}
ensures goal \eqref{eq2_20}.
\end{theorem}

\begin{remark}
\label{Rem2}
Note that the model \eqref{eq2_4} with the Hurwitz matrix $(I-B(LB)^{-1}L)A-TL$ can describe many technical and technological systems. For example, the control of distillation column \cite{Afanasiev96,Bia05}, where the control signal is the irrigation flow and the output signal is the composition of the light fractions of the column top; the aircraft control
\cite{Afanasiev96,Fradkov11} at various heights and Mach numbers, where $u$ is the control of elevators, $y$ is vertical acceleration; electric DC motor control \cite{Ruderman08}, where the control signal is the input voltage, the output signal is the angular velocity, etc.
\end{remark}

\begin{proof}
Taking into account \eqref{eq2_2} and \eqref{eq2_4}, rewrite expression \eqref{eq2_3} in the form
\begin{equation}
\label{eq2_5}
\begin{array}{l} 
\dot{\varepsilon}=\left(\frac{\partial\Phi(\varepsilon,t)}{\partial \varepsilon}\right)^{-1}
\left(LAx+LBu+\varphi\right),
\end{array}
\end{equation}
where $\varphi =LDf-\frac{\partial\Phi(\varepsilon,t)}{\partial t}$ is the bounded function w.r.t. $\varepsilon$ and $t$.
Substituting the control law \eqref{eq2_6} into the first equation of \eqref{eq2_4} and \eqref{eq2_5}, we get
\begin{align}
& \dot{x}=(A-B(LB)^{-1}LA-TL)x
\\
&~~~~~~ -KB(LB)^{-1}\varepsilon+Df+T\Phi(\varepsilon,t),
\label{eq2_8} \\
& \dot{\varepsilon}=\left(\frac{\partial\Phi(\varepsilon,t)}{\partial \varepsilon}\right)^{-1}
\left[-K\varepsilon+\varphi\right].
\label{eq2_8a}
\end{align}
Analyze equation \eqref{eq2_8a} on the ISS. To this end, choose Lyapunov function of the form $V=0.5\varepsilon^{2}$.
Substituting \eqref{eq2_8a} into the condition $\dot{V}+2\alpha V\left(\frac{\partial\Phi(\varepsilon,t)}{\partial \varepsilon}\right)^{-1}-\beta \varphi^2 \left(\frac{\partial\Phi(\varepsilon,t)}{\partial \varepsilon}\right)^{-1} \leq 0$, where $\alpha>0 $ and $\beta>0$, we get $-(K-\alpha)\varepsilon^2+\varepsilon \varphi-\beta \varphi^2 \leq 0$.
If LMI \eqref{LMI_Th02} holds, then the last inequality is satisfied and system \eqref{eq2_8a} is stable.
Consequently, the signal $\varepsilon(t)$ is bounded.
If the matrix $A-B(LB)^{-1}LA-TL$ is Hurwitz, then the boundedness of the signal $x(t)$ follows from the boundedness of the signals $\varepsilon(t)$, 
$\Phi(\varepsilon,t)$ and $f(t)$. Therefore, the control law $u(t)$ given by \eqref{eq2_6} is bounded. Taking into account Theorem \ref{Th1}, goal \eqref{eq2_2} is satisfied. Theorem \ref{Th02} is proved.
\end{proof}

\textit{Example 5.}
Let in \eqref{eq2_4} parameters are given in the forms
\begin{equation}
\label{eq2_10}
\begin{array}{l} 
A=\begin{bmatrix}
0 & 1\\
1 & 2
\end{bmatrix},~~
B=\begin{bmatrix}
0\\
1
\end{bmatrix},~~ 
D=\begin{bmatrix}
1\\
1
\end{bmatrix},~~
L=[1~2],
\\
x(0)=[2~1]^{\rm T},
~~
f(t)=0.1+\sin(3t)+sat\left(\frac{d(t)}{0.3}\right),
\end{array}
\end{equation}
where $sat(\cdot)$ is the saturation function, the signal $d(t)$ is simulated in Matlab Simulink by using the "Band-Limited White Noise" block with a noise power of 0.1 and a sampling time of 0.1. It is required to ensure that the output signal $y(t)$ belongs to the set $\underline{r}g(t)<y(t)<\overline{r}g(t)$, where $\underline{r}=0.8$ and $\overline{r}=1$, and the function $g(t)$ will be given below.

The matrix $A-B(LB)^{-1}LA-TL$ is Hurwitz, for example, for all $T=[T_1~T_2]^{\rm T}$, where $T_1>0$ and $T_2>0$. Choose $K=1$ in \eqref{eq2_6}.
Define the function $\Phi(\varepsilon,t)$ as in Example 2, where $g$ is given by
\begin{equation}
\label{eq2_9}
\begin{array}{l} 
g(t)=(g_0-g_{\infty})e^{-kt}+g_{\infty}.
\end{array}
\end{equation}
Here $g_0=y(0)+0.01$, $g_{\infty}=0.1$ and $k=0.5$.
Fig.~\ref{Fig1} shows the transients in $y(t)$, $u(t)$ and $f(t)$. The oscillations of the control signal in Fig.~\ref{Fig1},\textit{b} are caused by the presence of the disturbance $f$. Moreover, it follows from Fig.~\ref{Fig1},\textit{b} that after third second the magnitude of the control signal is comparable with the magnitude of the disturbance. Fig.~\ref{Fig2a} presents the simulations under $f=0$. Thus, the plant can be stabilized in a given set by a not large value of the control signal.
\begin{figure}[h!]
\begin{minipage}[h]{0.49\linewidth}
\center{\includegraphics[width=1\linewidth]{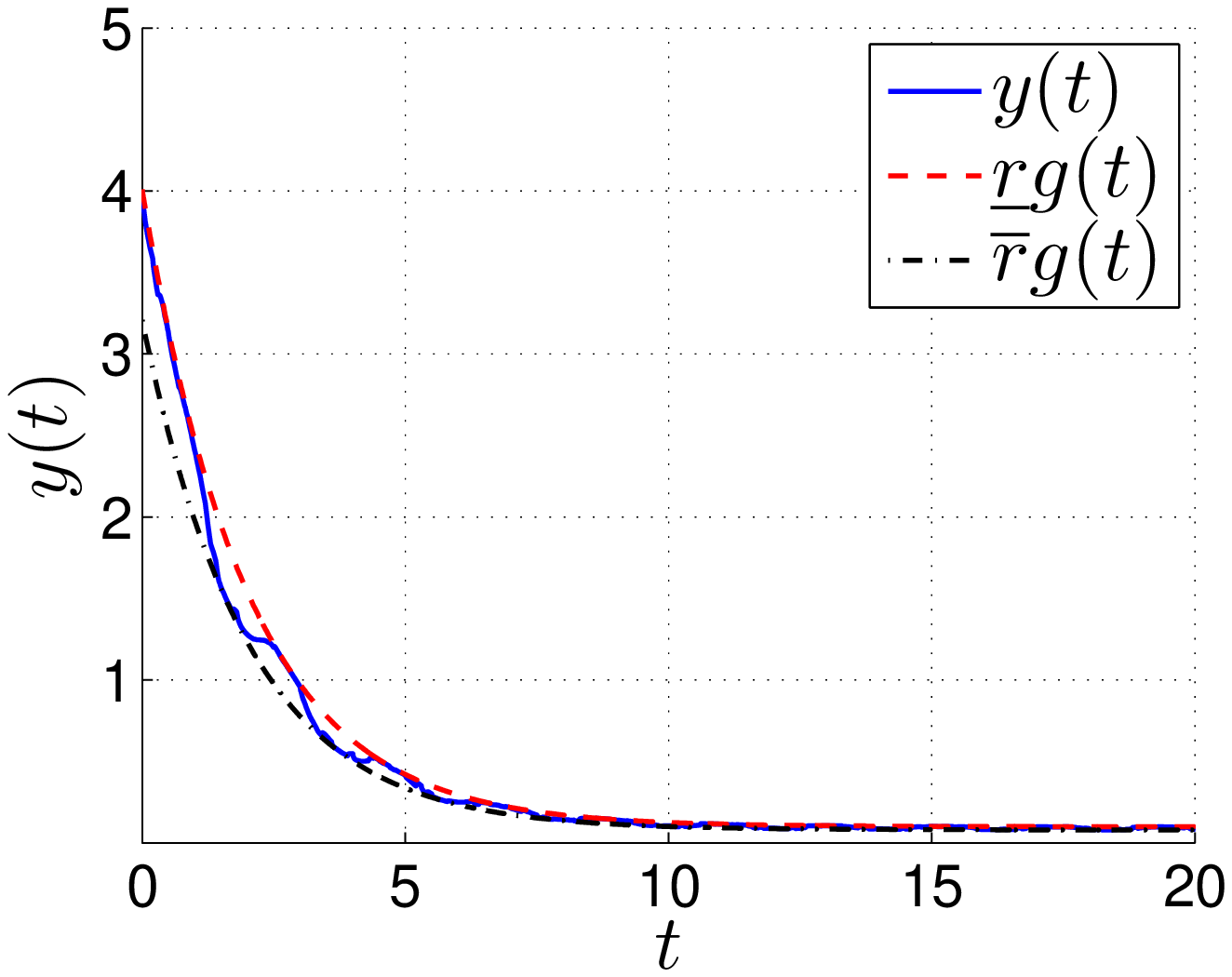}} \\ \textit{a}
\end{minipage}
\hfill
\begin{minipage}[h]{0.49\linewidth}
\center{\includegraphics[width=1\linewidth]{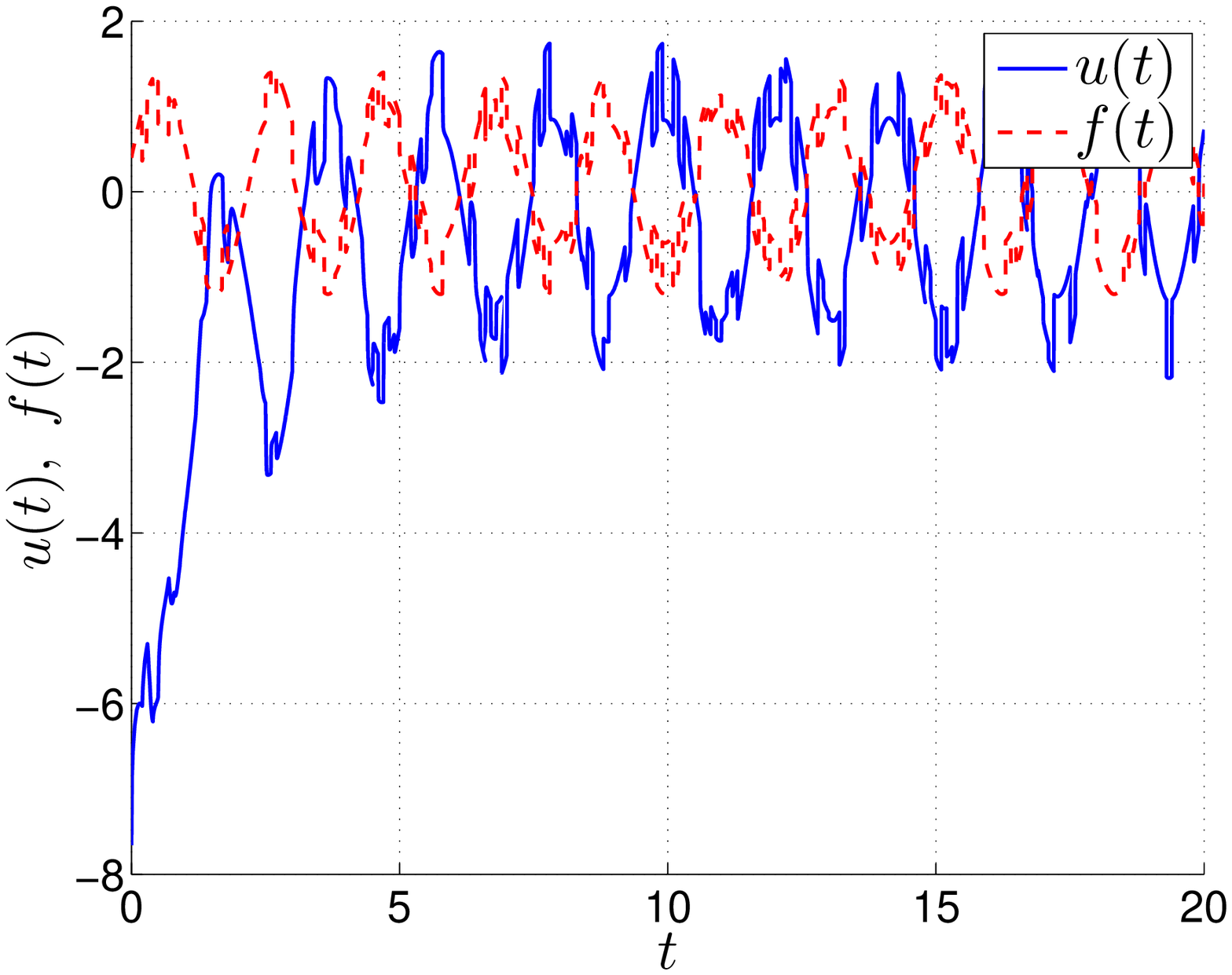}} \\ \textit{b}
\end{minipage}
\caption{The transients in $y(t)$ (\textit{a}), $u(t)$ è $f(t)$ (\textit{b}) for $g(t)$ given by \eqref{eq2_9}.}
\label{Fig1}
\end{figure}

\begin{figure}[h!]
\begin{minipage}[h]{0.49\linewidth}
\center{\includegraphics[width=1\linewidth]{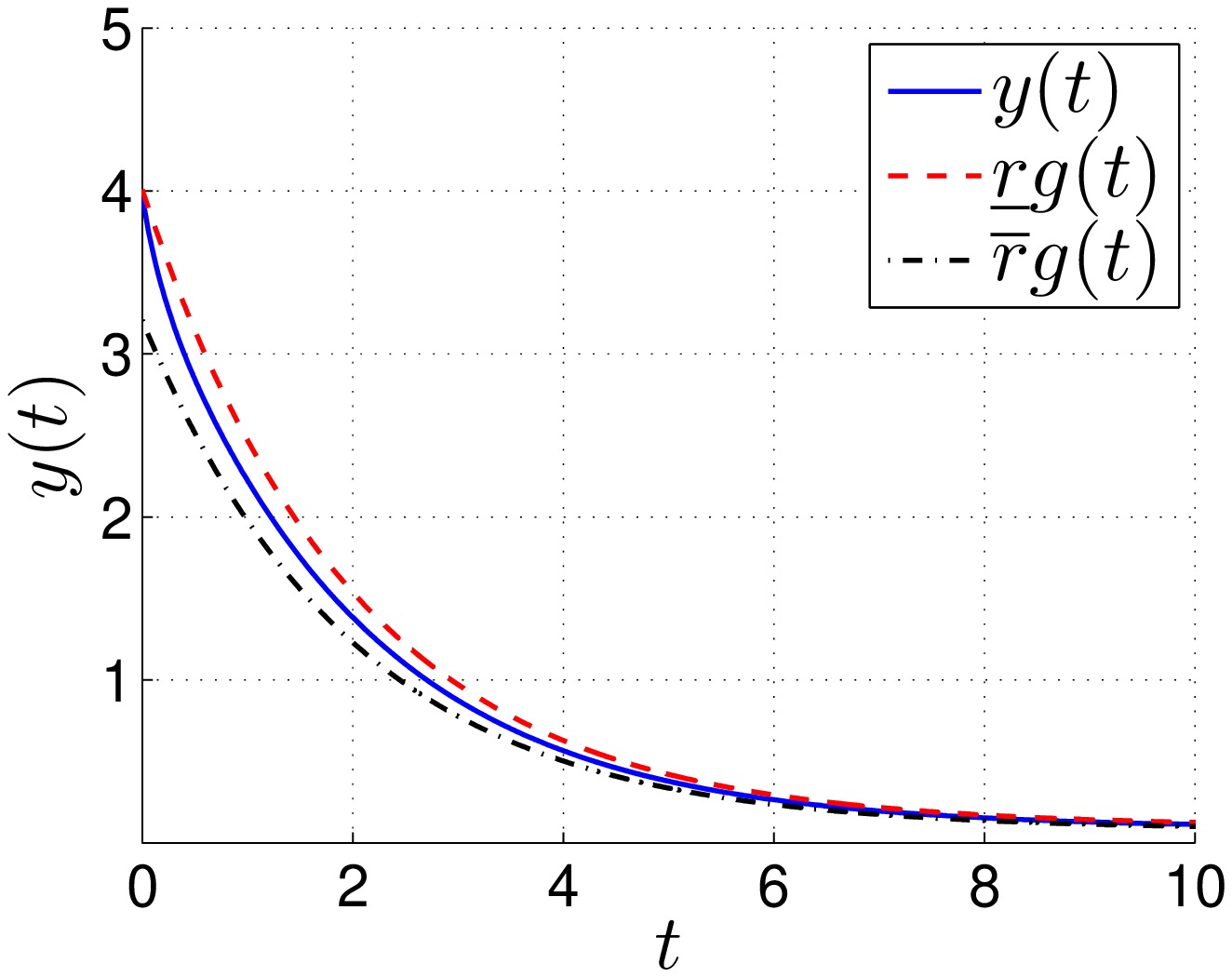}} \\ \textit{a}
\end{minipage}
\hfill
\begin{minipage}[h]{0.49\linewidth}
\center{\includegraphics[width=1\linewidth]{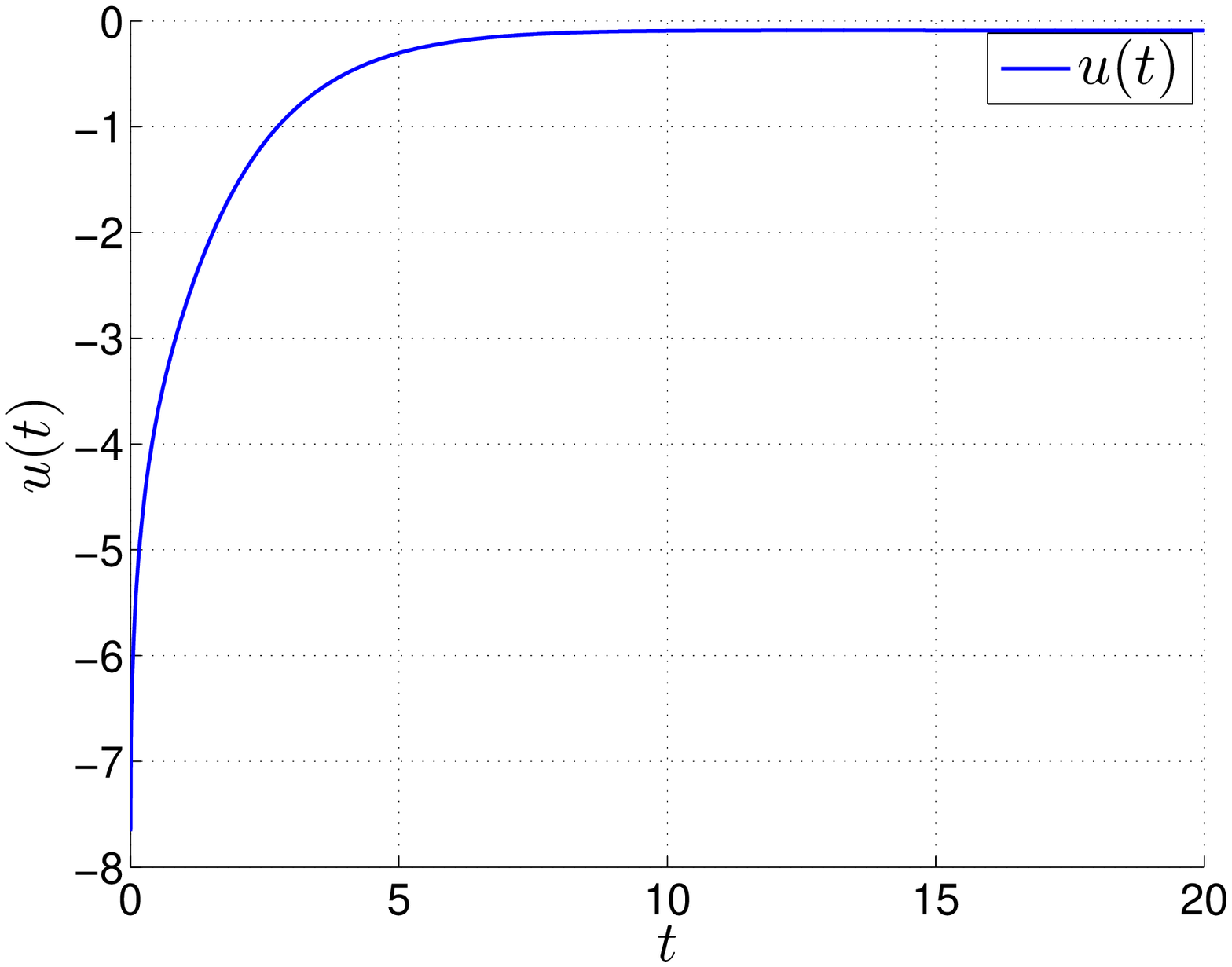}} \\ \textit{b}
\end{minipage}
\caption{The transients in $y(t)$ (\textit{a}) è $u(t)$ (\textit{b}) for $g(t)$ given by \eqref{eq2_11} for $f=0$.}
\label{Fig2a}
\end{figure}

Fig.~\ref{Fig2} shows the simulations for $y(t)$ and $u(t)$ for the set $0.8g(t)<y(t)<g(t)$, where the function $g(t)$ is given by
\begin{equation}
\label{eq2_11}
\begin{array}{l} 
g(t)=g_0\sin(kt)+g_0+g_{\infty}. 
\end{array}
\end{equation}

\begin{figure}[h!]
\begin{minipage}[h]{0.49\linewidth}
\center{\includegraphics[width=1\linewidth]{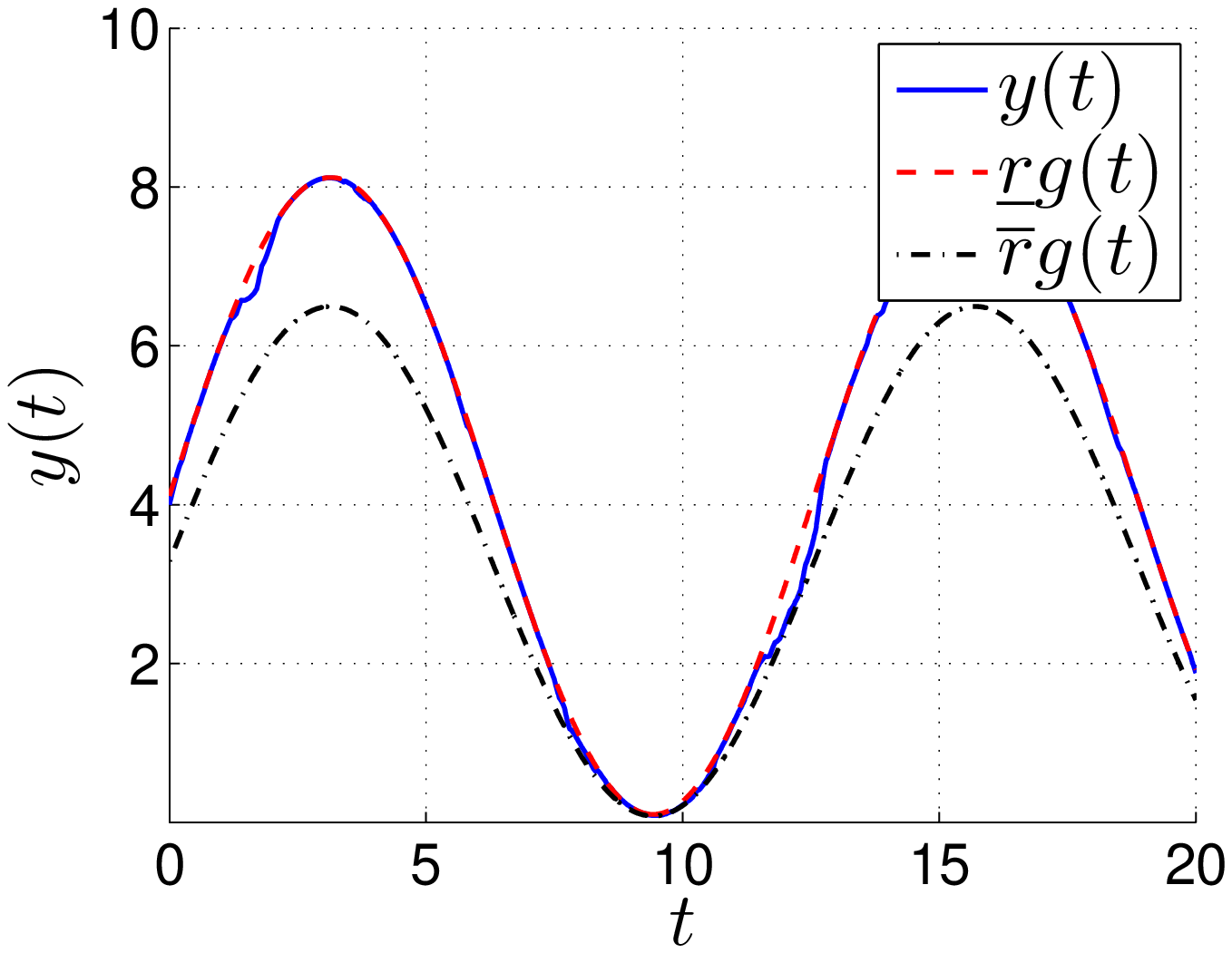}} \\ \textit{a}
\end{minipage}
\hfill
\begin{minipage}[h]{0.49\linewidth}
\center{\includegraphics[width=1\linewidth]{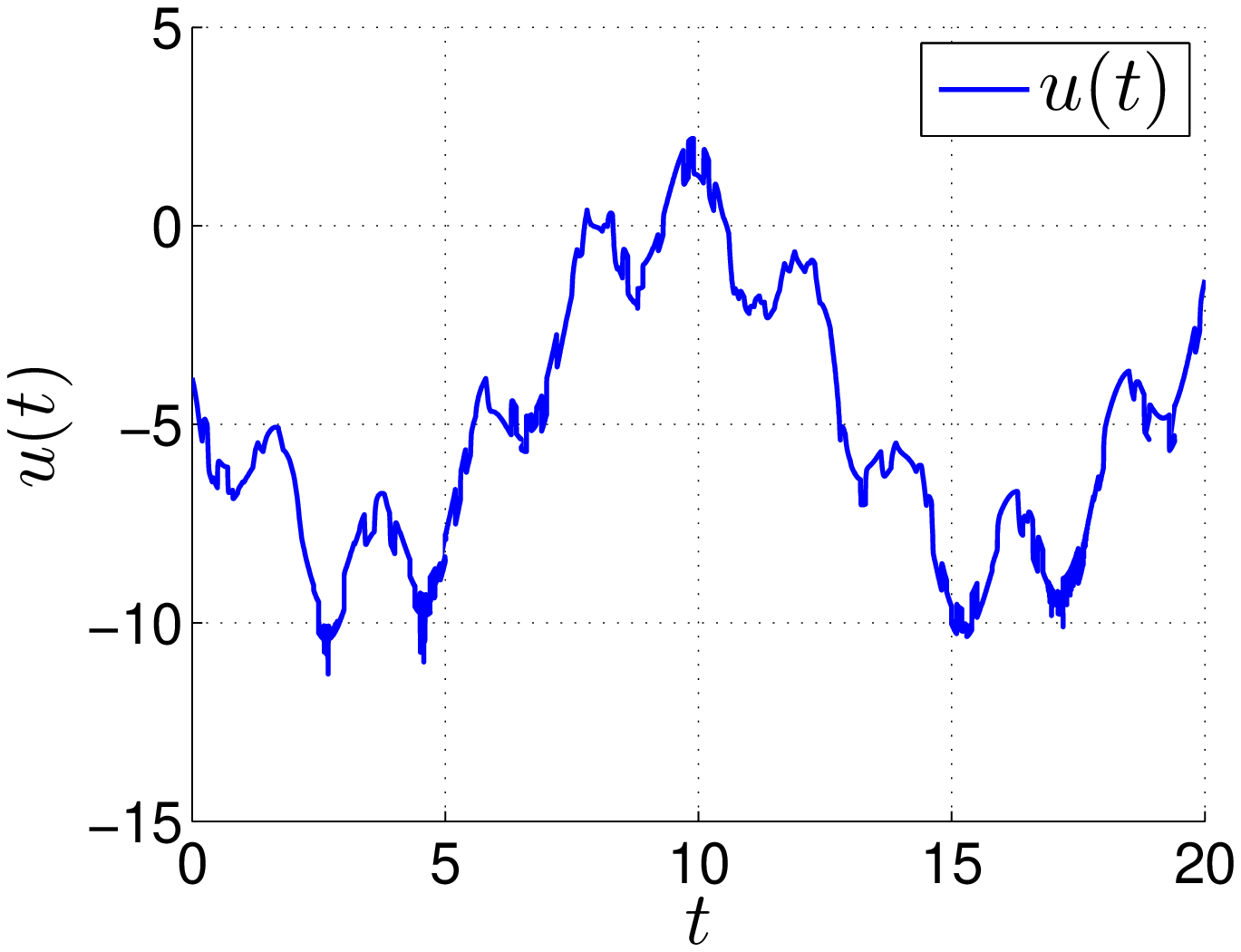}} \\ \textit{b}
\end{minipage}
\caption{The transients in $y(t)$ (\textit{a}) è $u(t)$ (\textit{b}) for $g(t)$ given by \eqref{eq2_11}.}
\label{Fig2}
\end{figure}

\section{Output feedback control for plants with sector nonlinearity and disturbances}
\label{Sec5}

Consider a plant model in the form

\begin{equation}
\label{eq5_1}
\begin{array}{l} 
\dot{x}=Ax+G\varphi(x,t)+Bu+Df,
\\
y=Lx.
\end{array}
\end{equation}
Here the state vector $x \in \mathbb R^n$ is unmeasured, $u \in \mathbb R^m$ and $y \in \mathbb R^v$ are measured signals, the disturbance $f \in \mathbb R^l$ is bounded signal. The matrices $A \in \mathbb R^{n \times n}$, $G \in \mathbb R^{n \times k}$, $B \in \mathbb R^{n \times m}$ and $L \in \mathbb R^{v \times n}$ are known and the matrix $D \in \mathbb R^{n \times l}$ is unknown. Unknown nonlinearity $ \varphi (x,t) \in \mathbb R ^ {k} $ satisfies the condition $|\varphi(x,t)| \leq C|x|$, $C>0$ is a known constant. The pair $(A,B)$ is controllable and the pair $(L,A)$ is observable.

Introduce the control law in the form
\begin{equation}
\label{eq5_2}
\begin{array}{l} 
u=K_1y+K_2\varepsilon,
\end{array}
\end{equation}
where $K_1 \in \mathbb R^{m \times v}$ and $K_2  \in \mathbb R^{m \times v}$ are chosen by the designer. In particular, $K_1$ and $K_2$ can be chosen such that the matrices $A+BK_1L$ and $LBK_2$ are Hurwitz.
Taking into account \eqref{eq2_2} and \eqref{eq5_2}, rewrite \eqref{eq2_3} and \eqref{eq5_1} in the forms
\begin{equation}
\label{eq5_3}
\begin{array}{l} 
\dot{x}=(A+BK_1L+T_1L)x+BK_2\varepsilon+G\varphi(x)
\\~~~~~+Df-T_1\Phi(\varepsilon,t),
\\
\dot{\varepsilon}=\left(\frac{\partial\Phi(\varepsilon,t)}{\partial \varepsilon}\right)^{-1}
\big[LBK_2\varepsilon+(LA+LBK_1L
\\
~~~~~+T_2L)x+LG\varphi(x)+LDf
\\
~~~~~-\frac{\partial\Phi(\varepsilon,t)}{\partial t}-T_2\Phi(\varepsilon,t)\big].
\end{array}
\end{equation}
Here $T_1 \in \mathbb R^{n \times v}$ and $T_2 \in \mathbb R^{v \times v}$. Introduce the following notation

\begin{equation}
\label{eq5_4}
\begin{array}{l}
x_e=col\{x,~\varepsilon\},~~~f_e=col\left\{f,~\frac{\partial\Phi(\varepsilon,t)}{\partial t},~\Phi(\varepsilon,t)\right\},
\\
A_{21}(\varepsilon,t)=\left(\frac{\partial\Phi(\varepsilon,t)}{\partial \varepsilon}\right)^{-1}(LA+LBK_1L+T_2L),
\\
A_{22}(\varepsilon,t)=\left(\frac{\partial\Phi(\varepsilon,t)}{\partial \varepsilon}\right)^{-1}LBK_2,
\\
A_e(\varepsilon,t)=\begin{bmatrix}
A+BK_1L+T_1L & BK_2\\
A_{21} & A_{22}
\end{bmatrix},
\\
G_e(\varepsilon,t)=\begin{bmatrix}
G\\
\left(\frac{\partial\Phi(\varepsilon,t)}{\partial \varepsilon}\right)^{-1}LG
\end{bmatrix},
\\
D_{21}(\varepsilon,t)=\left(\frac{\partial\Phi(\varepsilon,t)}{\partial \varepsilon}\right)^{-1}LD,
\\
D_{22}(\varepsilon,t)=-\left(\frac{\partial\Phi(\varepsilon,t)}{\partial \varepsilon}\right)^{-1},
\\
D_{23}(\varepsilon,t)=-\left(\frac{\partial\Phi(\varepsilon,t)}{\partial \varepsilon}\right)^{-1}T_2,
\\
D_e(\varepsilon,t)=\begin{bmatrix}
D & 0 & -T_1\\
D_{21}(\varepsilon,t) & D_{22}(\varepsilon,t) & D_{23}(\varepsilon,t)
\end{bmatrix}.
\end{array}
\end{equation}

Considering \eqref{eq5_4}, rewrite \eqref{eq5_3} as follows
\begin{equation}
\label{eq5_5}
\begin{array}{l} 
\dot{x}_e=A_e(\varepsilon,t) x_e+G_e(\varepsilon,t) \varphi(x,t)+D_e(\varepsilon,t) f_e.
\end{array}
\end{equation}

\begin{theorem}
\label{Th2}
Let conditions (a)-(d) hold for transformation \eqref{eq2_2},
$\frac{\partial\Phi(\varepsilon,t)}{\partial \varepsilon}>0$ for all $\varepsilon$ and $t$. 
Given $\alpha>0$, $K_1$, $K_2$, $T_1$ and $T_2$ there exist the coefficient $\beta>0$ and the matrix $P=P^{\rm T}>0 $ such that the following matrix inequality holds
\begin{equation}
\label{eq5_Th2}
\begin{array}{l}
\begin{bmatrix}
 \Psi_{11}(\varepsilon,t) & PG_e(\varepsilon,t) & PD_e(\varepsilon,t)\\
* & -1 & 0\\
* & * & -\beta I
\end{bmatrix}
\leq 0.
\end{array}
\end{equation}
Here $"*"$ defines the symmetric block of the symmetric matrix, $E=[I~0]$, $\Psi_{11}(\varepsilon,t)=A_e(\varepsilon,t)^{\rm T}P+PA_e(\varepsilon,t)+\alpha P+C^2 E^{\rm T}E$. Then control law \eqref{eq5_2} ensures goal \eqref{eq2_20}.
\end{theorem}

\begin{proof}
For the ISS analysis of \eqref{eq5_5} consider Lyapunov function in the form
$V=x_e^{\rm T}Px_e$. Considering \eqref{eq5_5} and substituting the expression for $V$ in the inequality
\begin{equation}
\label{eq5_7}
\begin{array}{l} 
\dot{V}+\alpha V-\beta f_e^{\rm T} f_e \leq 0,
\end{array}
\end{equation}
we get
\begin{equation}
\label{eq5_8}
\begin{array}{l} 
x_e^{\rm T}[A_e(\varepsilon,t)^{\rm T}P+PA_e(\varepsilon,t)+\alpha P]x_e
-\beta f_e^{\rm T} f_e
\\
+2x_e^{\rm T}PG_e(\varepsilon,t)\varphi(x,t)
+2x_e^{\rm T}PD_e(\varepsilon,t)f_e
 \leq 0.
\end{array}
\end{equation}

Introduce the new vector $z=col\{x_e,\varphi(x,t),f_e\}$ and rewrite inequality \eqref{eq5_8} as
\begin{equation}
\label{eq5_9}
\begin{array}{l} 
z^{\rm T}
\begin{bmatrix}
\Psi_{11}(\varepsilon,t)-C^2E^{\rm T}E & PG_e(\varepsilon,t) & PD_e(\varepsilon,t)\\
* & 0 & 0\\
* & * & -\beta I
\end{bmatrix}
z \leq 0.
\end{array}
\end{equation}

Rewrite the inequality $\varphi^2(x,t) \leq C^2 x_e^{\rm T} E^{\rm T}E x_e$ in the form
\begin{equation}
\label{eq5_10}
\begin{array}{l} 
z^{\rm T}
\begin{bmatrix}
C^2E^{\rm T}E & 0 & 0\\
* & -1 & 0\\
* & * & 0
\end{bmatrix}
z \geq 0.
\end{array}
\end{equation}
According to the S-procedure, inequalities \eqref{eq5_9} and \eqref{eq5_10} are simultaneously satisfied if inequality \eqref{eq5_Th2} holds. Therefore, the function $x_e(t)$ is bounded from \eqref{eq5_7}. Thus, the signals $x(t)$ and $\varepsilon(t)$ are bounded. Then control law \eqref{eq5_2} is bounded. Tacking into account Theorem \ref{Th1}, goal \eqref{eq2_2} is satisfied. Theorem \ref{Th2} is proved.
\end{proof}

\textit{Example 6.} 
Let in \eqref{eq5_1}
\begin{equation*}
\begin{array}{l} 
A=\begin{bmatrix}
0 & 1 & 0\\
0 & 0 & 1\\
0.1 & -2 & -3
\end{bmatrix},~~
B=\begin{bmatrix}
1 & 2\\
1 & 1\\
1 & 2
\end{bmatrix},~~ 
D=\begin{bmatrix}
1\\
1\\
1
\end{bmatrix},~~
\\
L=\begin{bmatrix}
2 & 1 & 1\\
1 & 2 & 1\\
\end{bmatrix},
~~
G=
\begin{bmatrix}
0 & 0 & 0\\
0 & 0 & 0\\
0.1 & 0.1 & 0.1
\end{bmatrix},~~
\varphi(x)=sin(x),
\end{array}
\end{equation*}
the disturbance $f(t)$ is given by \eqref{eq2_10}.

Choose
$K_1=0.01\begin{bmatrix}
0 & 0\\
-1 & -1
\end{bmatrix}$, $K_2=\begin{bmatrix}
1.5 & -1.75\\
-1 & 1
\end{bmatrix}$
in control law \eqref{eq5_2}.
Additionally, choose $T_1=\begin{bmatrix}
1 & 2 & 1\\
1 & 2 & 1
\end{bmatrix}$ è $T_2=\begin{bmatrix}
1 & 2\\
1 & 2
\end{bmatrix}$.

Let $\Phi(\varepsilon,t)=diag\{\Phi_1(\varepsilon_1,t),\Phi_2(\varepsilon_2,t)\}$, where $\Phi_i$ is given in example 3: $\Phi_i(\varepsilon_i,t)=\frac{\overline{g}_i(t)e^{\varepsilon_i}+\underline{g}_i(t)}{e^{\varepsilon_i}+1}$, $i=1,2$. Therefore, $\Phi(\varepsilon,t)>0$ for all $ \varepsilon$ and $t$.
Then
$\frac{\partial \Phi_i(\varepsilon_i,t)}{\partial \varepsilon_i}=\frac{e^{\varepsilon_i}(\overline{g}_i(t)-\underline{g}_i(t))}{(e^{\varepsilon_i}+1)^2}>0$ since $\overline{g}_i(t)>\underline{g}_i(t)$.
Additionally, $\left(\frac{\partial \Phi_i(\varepsilon_i,t)}{\partial \varepsilon_i}\right)^{-1} \to + \infty$ at $\varepsilon \to +\infty$ and the smallest value of $\left(\frac{\partial \Phi_i(\varepsilon_i,t)}{\partial \varepsilon_i}\right)^{-1} \Big|_{\varepsilon=0} = \frac{4}{\overline{g}_i(t)-\underline{g}_i(t)}>0$.

According to \cite{Fridman10}, if LMI is feasible on the vertices of a polytope, then inside the polytope LMI also is feasible.
In our case for every fixed $\frac{\partial \Phi_i(\varepsilon_i,t)}{\partial \varepsilon_i}$ the matrix inequality \eqref{eq5_Th2} is linear.
However, the polytop is unbounded, since $\left(\frac{\partial \Phi_i(\varepsilon_i,t)}{\partial \varepsilon_i}\right)^{-1} \to + \infty$  at $\varepsilon \to +\infty$. The simulations with increasing $\left(\frac{\partial \Phi_i(\varepsilon_i,t)}{\partial \varepsilon_i}\right)^{-1}$ show that the eigenvalues of the matrix $P$ converge to some positive values. At the vertices 
$\frac{4}{\overline{g}_i(t)-\underline{g}_i(t)}$ the matrix inequality \eqref{eq5_Th2} holds too.

Choose the parameters of the function $\Phi(\varepsilon,t)$ in the form
\begin{equation}
\label{eq5_11}
\begin{array}{l} 
\overline{g}_1(t)=(g_{0}-g_{1})e^{-kt}+g_{1},
\\
\overline{g}_2(t)=(g_0-g_{2})\cos(kt)+g_{4},
\\
\underline{g}_1(t)=(g_0-g_{2})e^{-kt}+g_{3},
\\
\underline{g}_2(t)=\cos(kt)+g_{5},
\end{array}
\end{equation}
where $g_0=\sqrt{y^{\rm T}(0)y(0)}+0.01$, $g_{1}=0.1$, $g_{2}=2$, $g_{3}=-0.2$, $g_{4}=g_0-0.1$, $g_{5}=0.8$ and $k=0.5$.
Fig.~\ref{Fig3},~\ref{Fig3a} show the transients in $y_1(t)$ , $y_2(t)$ and $u(t)=col\{u_1(t),u_2(t)\}$ for $x(0)=col\{\frac{5}{3},\frac{2}{3},-1\}$.

\begin{figure}[h]
\begin{minipage}[h]{0.49\linewidth}
\center{\includegraphics[width=1\linewidth]{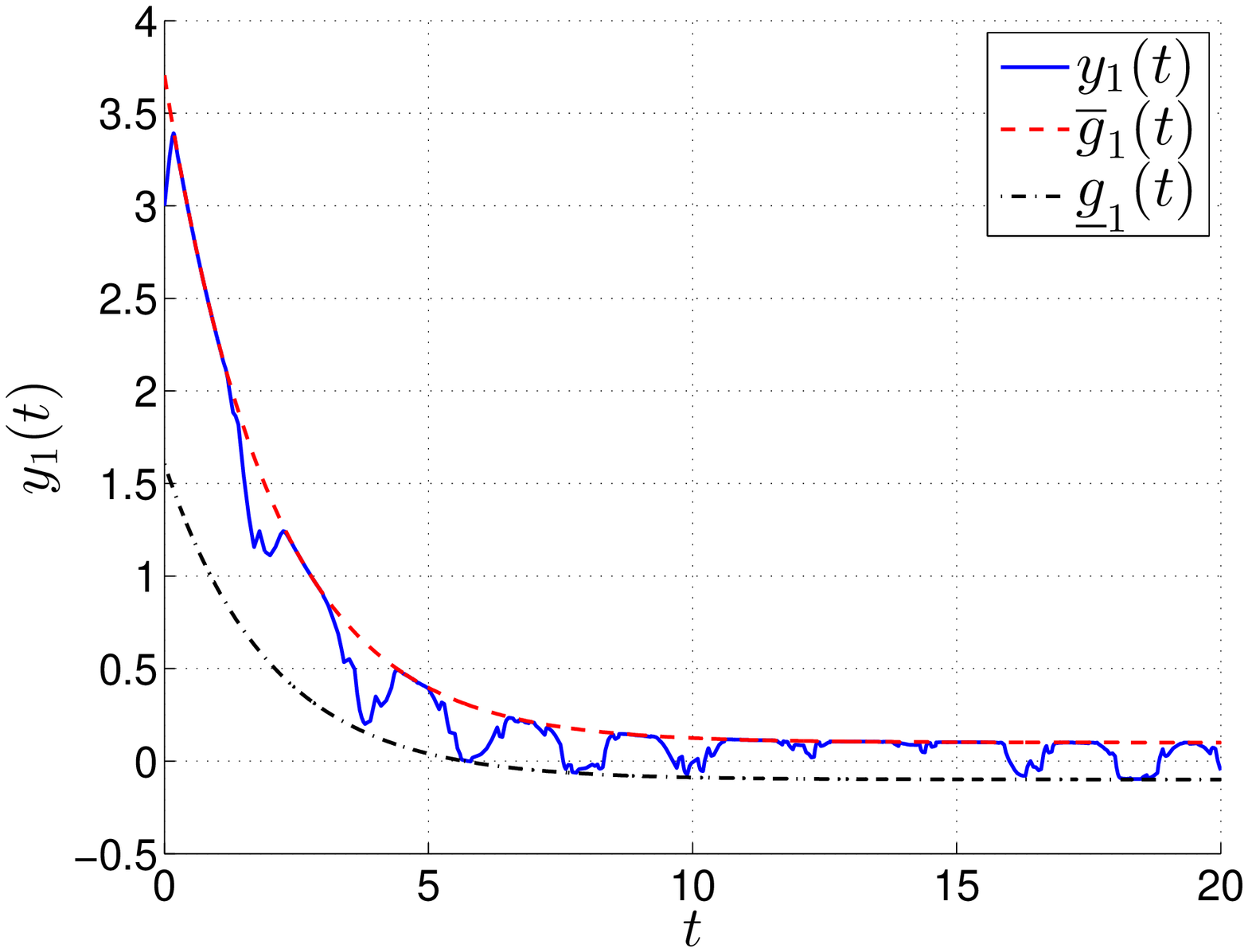}}
\end{minipage}
\hfill
\begin{minipage}[h]{0.49\linewidth}
\center{\includegraphics[width=1\linewidth]{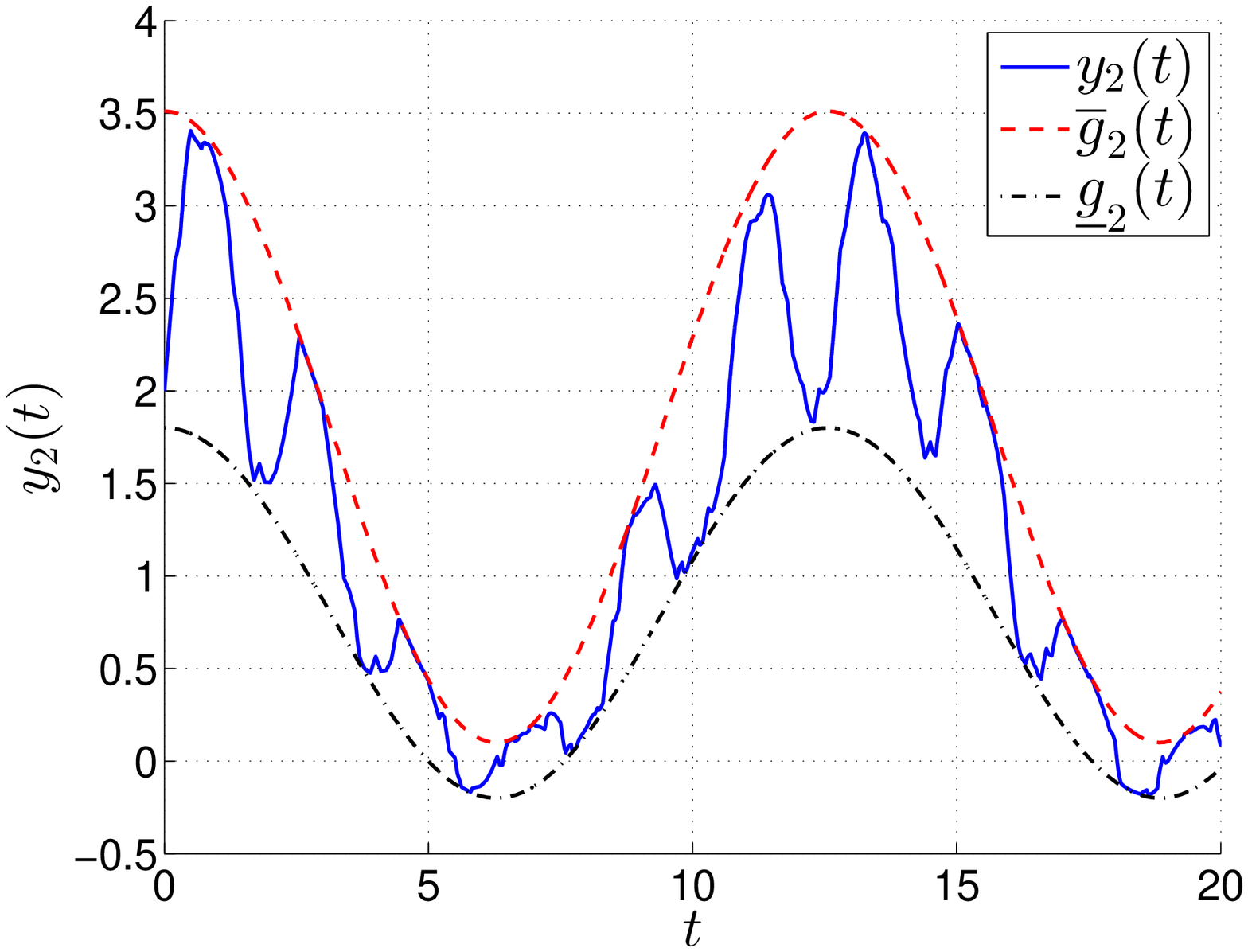}}
\end{minipage}
\caption{The transients in $y_1(t)$ and $y_2(t)$ for $\Phi(\varepsilon,t)$ with \eqref{eq5_11} and $x(0)=col\{\frac{5}{3},\frac{2}{3},-1\}$.}
\label{Fig3}
\end{figure}

\begin{figure}[h]
\center{\includegraphics[width=0.5\linewidth]{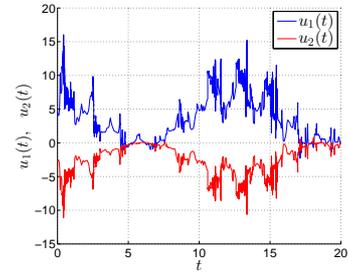}}
\caption{The transients in $u(t)=col\{u_1(t),u_2(t)\}$ for $\Phi(\varepsilon,t)$ with  \eqref{eq5_11} and $x(0)=col\{\frac{5}{3},\frac{2}{3},-1\}$.}
\label{Fig3a}
\end{figure}

Note that the control law $u=K_1y+K_2\varepsilon$ does not depend on the parameters of plant \eqref{eq2_4}. The simulations show the proposed control low is robust under unknown parameters of \eqref{eq2_4}. Thus, the closed-loop system remains stable for
$A=\begin{bmatrix}
0 & 1 & 0\\
0 & 0 & 1\\
a_1 & a_2 & a_3
\end{bmatrix}$ 
and 
$G=\begin{bmatrix}
0 & 0 & 0\\
0 & 0 & 0\\
g_{\varphi 1} & g_{\varphi 2} & g_{\varphi 3}
\end{bmatrix}$, 
where $a_1 \in [-5;0.1]$, $a_2 \in [-5;-2]$, $a_3 \in [-5;-3]$, $b \in [0.5; 10]$ è $g_{\varphi 1} \in [-3; 3]$, $g_{\varphi 2} \in [-3; 3]$ and $g_{\varphi 3} \in [-3; 3]$.

According to \eqref{eq2_2} and (a), the initial value $y(0)$ must belong to the sets
$\underline{g}_i(0)<y_i(0)<\overline{g}_i(0)$, $i=1,2$.
If the initial conditions have significant uncertainty, then the functions 
$\underline{g}_i(t)$ and $\overline{g}_i(t)$ can be specified with a margin at the initial time. For example, the functions $\underline{g}_i$ and $\overline{g}_i$ can be presented in the form
\begin{equation}
\label{eq5_12}
\begin{array}{l} 
\overline{g}_1(t)=(g_{0}-g_{1})e^{-kt}+g_{1}+g_6e^{-k_0t},
\\
\overline{g}_2(t)=(g_0-g_{2})\cos(kt)+g_{4}+g_6e^{-k_0t},
\\
\underline{g}_1(t)=(g_0-g_{2})e^{-kt}+g_{3}-g_6e^{-k_0t},
\\
\underline{g}_2(t)=\cos(kt)+g_5-g_6e^{-k_0t},
\end{array}
\end{equation}
where $g_6=3$ and $k_0=2$.
Fig.~\ref{Fig4} illustrates the plots of the output signals $y_1(t)$ and $y_2(t)$ with $x(0)=col\{\frac{10}{3},-\frac{5}{3},-1\}$.

\begin{figure}[h]
\begin{minipage}[h]{0.49\linewidth}
\center{\includegraphics[width=1\linewidth]{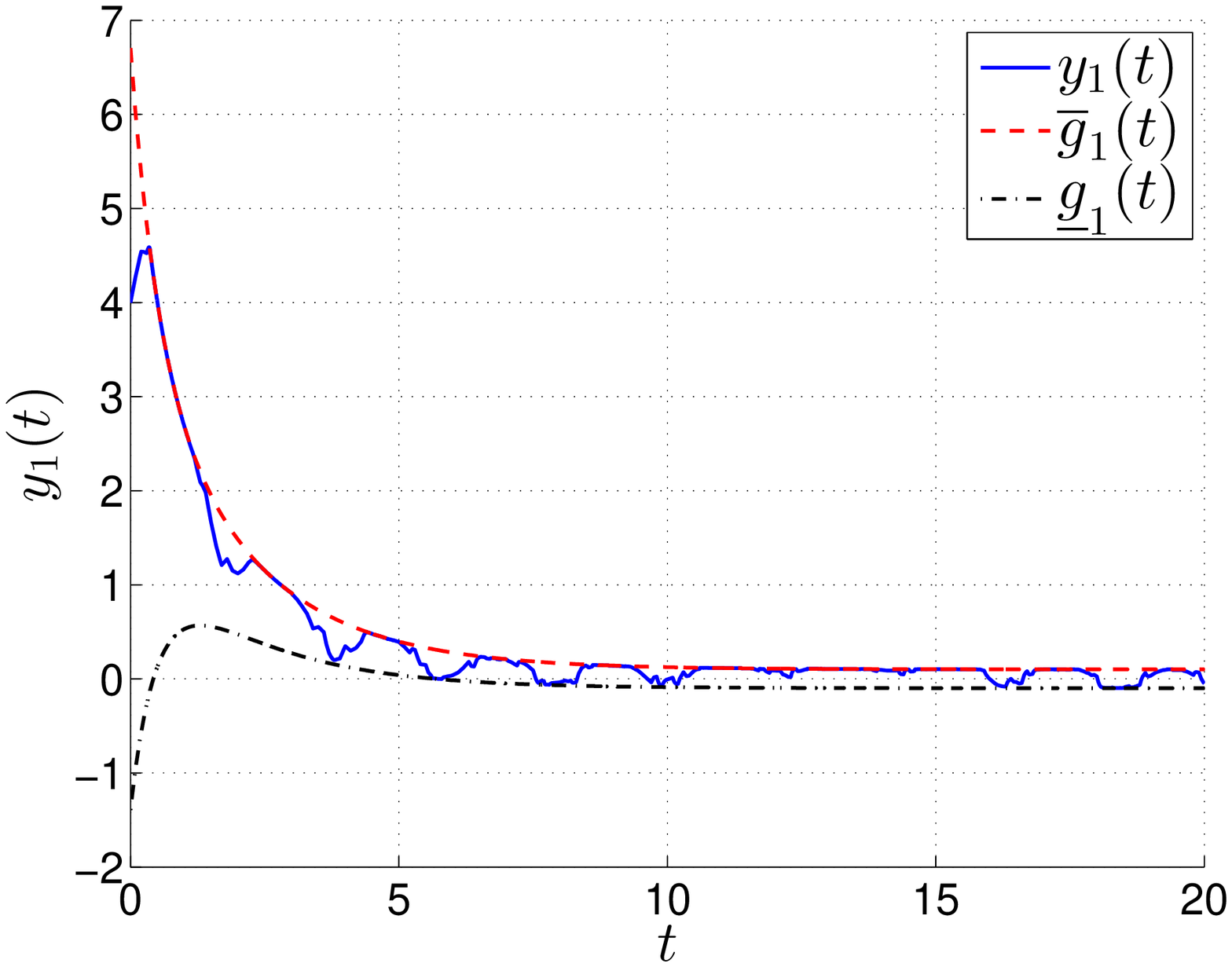}}
\end{minipage}
\hfill
\begin{minipage}[h]{0.49\linewidth}
\center{\includegraphics[width=1\linewidth]{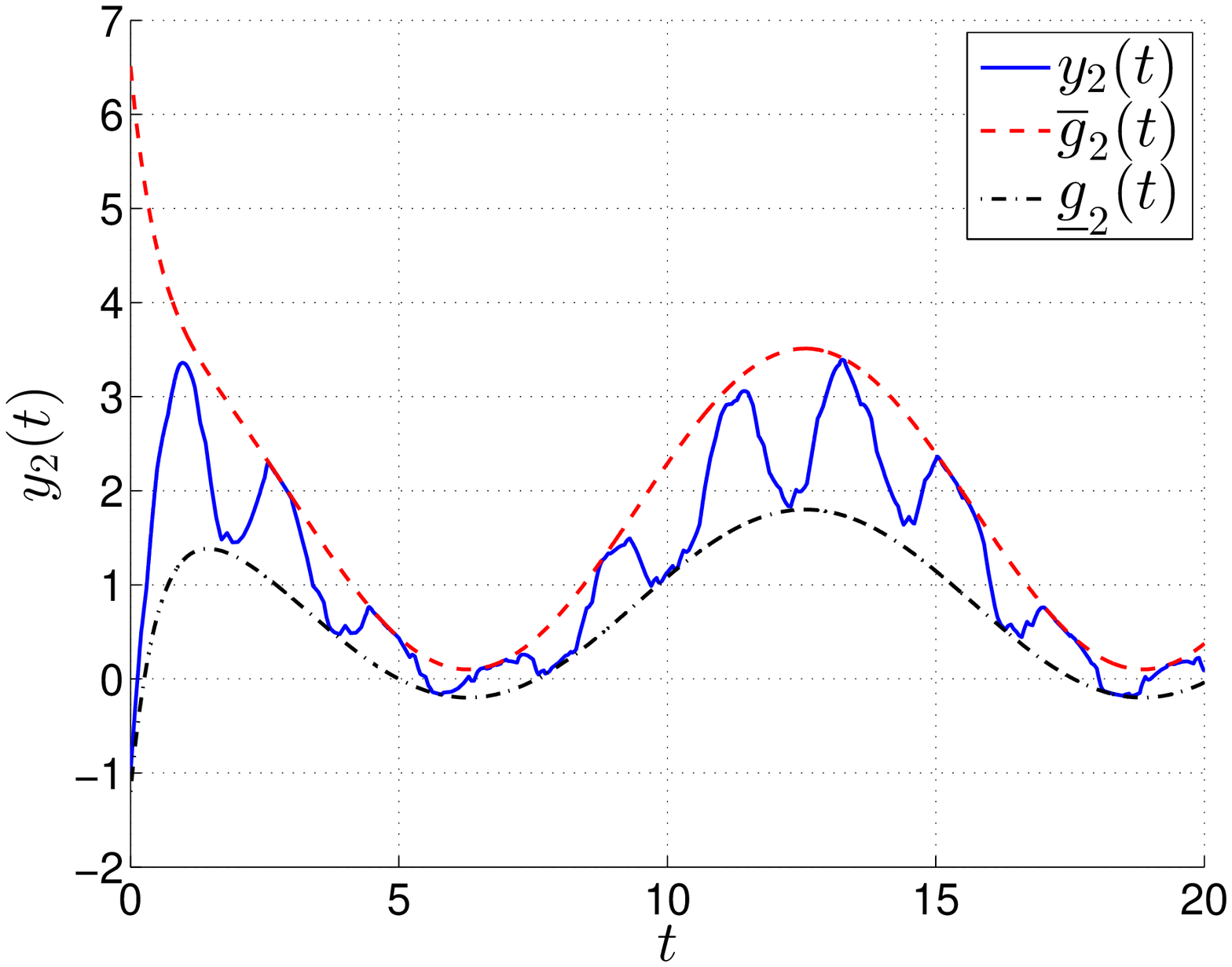}}
\end{minipage}
\caption{The transients in $y_1(t)$ and $y_2(t)$ for $\Phi(\varepsilon,t)$ with \eqref{eq5_12} and $x(0)=col\{\frac{10}{3},-\frac{5}{3},-1\}$.}
\label{Fig4}
\end{figure}

The simulations show that the transients in $y$ can be close to the boundaries of $\overline{g}(t)$ and $\underline{g}(t)$.
From $\varepsilon=\ln\frac{\underline{g}-y}{y-\overline{g}}$ it follows that the value of $|\varepsilon(t)|$ can take large values. Therefore, the computational load of the controller is increased. As a result, Matlab work is increased and sometimes Matlab gives an error in the calculations. To prevent this problem, it is recommended to select the parameters of the loop of $\varepsilon$ more than the parameters of the loop of $y$. Thereby, the transient time in $\varepsilon(t)$ is reduced in comparison with the transient time for $y(t)$. Moreover, it increases robustness w.r.t. uncertainty of plant parameters and the large value of the disturbance $f$. Let us demonstrate this fact. Rewrite the control law as $u=K_1y+\gamma K_2\varepsilon$, $\gamma>0$. Increasing $\gamma$, the transients in $y$ keep away from the boundaries $\overline{g}(t)$ and $\underline{g}(t)$ (see Fig.~\ref{Fig5}).

\begin{figure}[h]
\begin{minipage}[h]{0.49\linewidth}
\center{\includegraphics[width=1\linewidth]{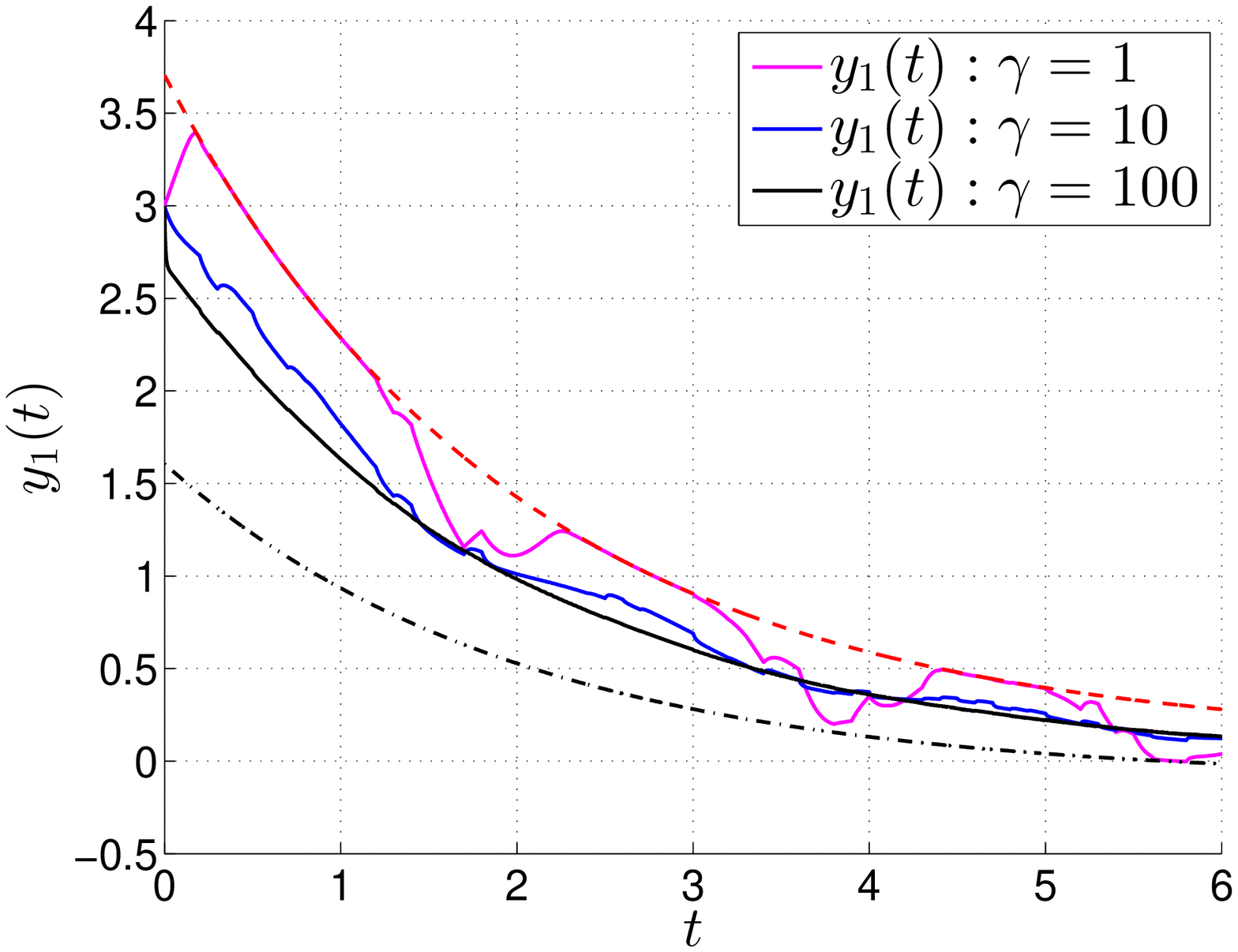}}
\end{minipage}
\hfill
\begin{minipage}[h]{0.49\linewidth}
\center{\includegraphics[width=1\linewidth]{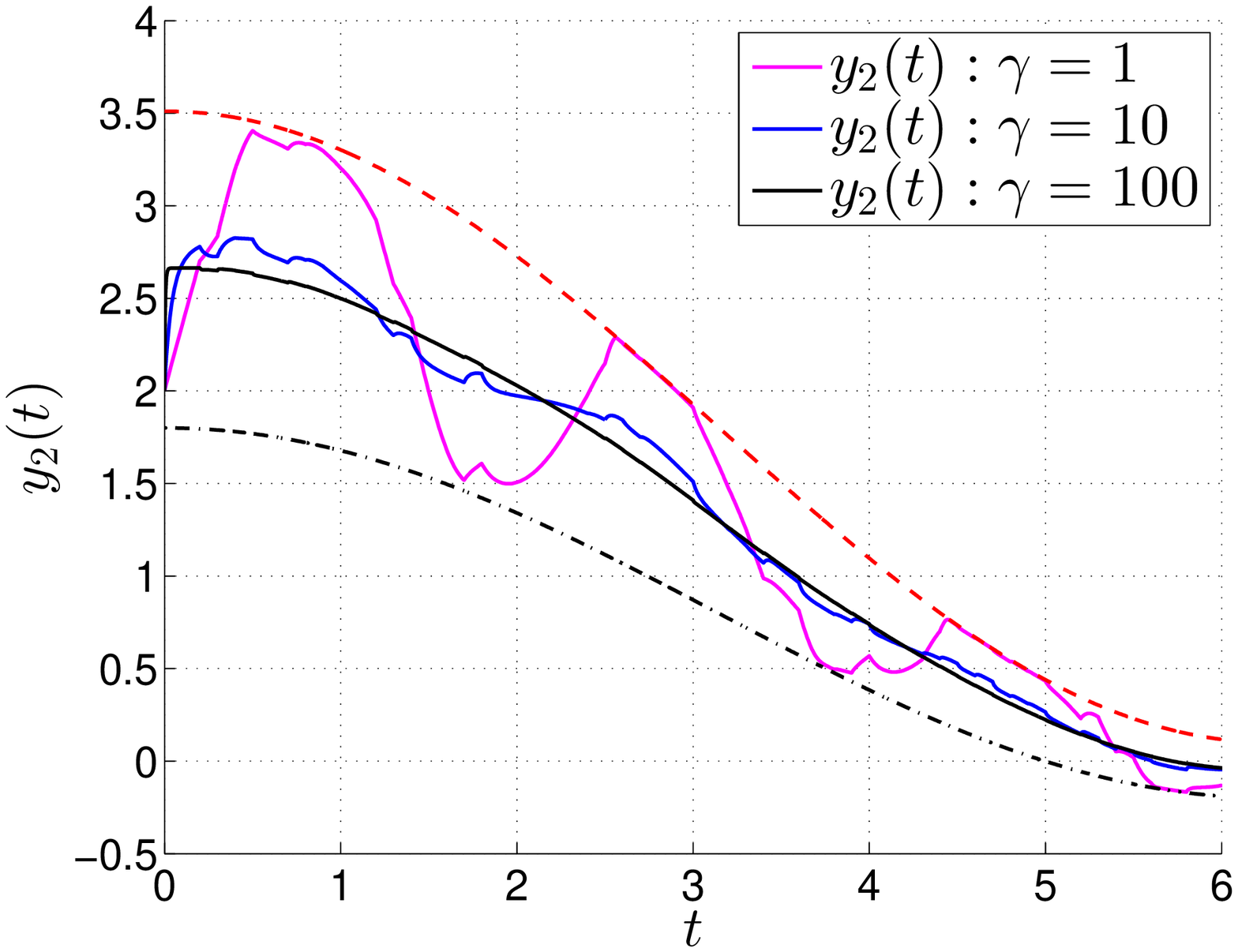}}
\end{minipage}
\caption{The transients in $y_1(t)$ and $y_2(t)$ for $\gamma=1$, $\gamma=10$ and $\gamma=100$ and $x(0)=col\{1~1~0\}$.}
\label{Fig5}
\end{figure}

\section{Output feedback control for systems with arbitrary relative degree}
\label{Sec6}

The results of sections \ref{Sec4} and \ref{Sec5} are valid for systems with relative degree that not exceeding one. Let us consider systems with arbitrary relative degree and they be described by the following equations
\begin{equation}
\label{eq6_2}
\begin{array}{l} 
\dot{x}=Ax+Bu+Df,
\\
y=Lx.
\end{array}
\end{equation}
Here $x \in \mathbb R^n$ is the unmeasured state vector, $u \in \mathbb R$ and $y \in \mathbb R$ are measured signals, the disturbance signal $f \in \mathbb R$ is bounded along with the $n-1$ derivatives. The matrices $A \in \mathbb R^{n \times n}$, $B \in \mathbb R^{n}$ and $L \in \mathbb R^{1 \times n}$ are known, the matrix $D \in \mathbb R^{n}$ is unknown and the matrix $A$ is Hurwitz. The pair $(A,B)$ is controllable and the pair $(L,A)$ is observable.

Transform equation \eqref{eq6_2} to the form
\begin{equation}
\label{eq6_3}
\begin{array}{l} 
y(t)=\frac{R(p)}{Q(p)}u(t)+\frac{D(p)}{Q(p)}f(t)+\epsilon(t).
\end{array}
\end{equation}
The linear differential operators $Q(p)$, $R(p)$ and $D(p)$ are obtained from transition \eqref{eq6_2} to \eqref{eq6_3}, i.e. $Q(p)=det(pI-A)$, $R(p)=L(pI-A)^*B$, $D(p)=L(pI-A)^*D$, $\epsilon(t)=L(pI-A)^*x(0)$ is the exponentially decaying function, $p=\frac{d}{dt}$. 
Substituting $y(t)$ from \eqref{eq6_3} into \eqref{eq2_3}, we get
\begin{equation}
\label{eq6_4}
\begin{array}{l}
\dot{\varepsilon}(t)=\Big(\frac{\partial\Phi(\varepsilon,t)}{\partial \varepsilon}\Big)^{-1}
\Big(\frac{pR(p)}{Q(p)}u(t)
\\
~~~~~+\frac{pD(p)}{Q(p)}f(t)+\dot{\epsilon}(t)-\frac{\partial\Phi(\varepsilon,t)}{\partial t}
\Big).
\end{array}
\end{equation}

Introduce the control law $u(t)=-\frac{Q(p)}{pR(p)}K\varepsilon(t)$, where $K>0 $ and rewrite equation \eqref{eq6_4} in the form $\dot{\varepsilon}=\Big(\frac{\partial\Phi(\varepsilon,t)}{\partial \varepsilon}\Big)^{-1}
\times \\ \Big(-K\varepsilon+\frac{pD(p)}{Q(p)}f(t)+\dot{\epsilon}(t)-\frac{\partial\Phi(\varepsilon,t)}{\partial t}\Big)$.
 This system is ISS when $\frac{\partial\Phi(\varepsilon,t)}{\partial \varepsilon}>0$. However, such control law is not implementable, because $\rho-1$ derivatives of the signal $\varepsilon$ are required for measurement, where $\rho=\deg Q(p)-\deg R(p)$ is the relative degree of \eqref{eq6_2}. Therefore, introduce the control law in the form
\begin{equation}
\label{eq6_5}
\begin{array}{l} 
u=-\frac{Q(p)}{R(p)[p(\mu p+1)^{\rho-1}+a\mu]}K\varepsilon.
\end{array}
\end{equation}
Here, the sufficiently small number $\mu>0$ and the coefficient $a>0$ are chosen such that the polynomial $\lambda(\mu \lambda+1)^{\rho-1}+a\mu$ is Hurwitz, $\lambda$ is a complex variable. Given \eqref{eq6_5}, rewrite \eqref{eq6_3} and \eqref{eq6_4} as follows
\begin{align}
& y(t)=-K\frac{1}{p(\mu p+1)^{\rho-1}+a\mu} 
\varepsilon(t)+\frac{D(p)}{Q(p)}f(t)+\epsilon(t),
\label{eq6_6} \\
& \dot{\varepsilon}(t)=\left(\frac{\partial\Phi(\varepsilon,t)}{\partial \varepsilon}\right)^{-1} \times \nonumber
\\
&\left(-K\varepsilon(t)-K\frac{p-p(\mu p+1)^{\rho-1}-a\mu}{p(\mu p+1)^{\rho-1}+a\mu}\varepsilon(t)
+\phi(t)
\right),
\label{eq6_60}
\end{align}
where $\phi(t)=\frac{pD(p)}{Q(p)}f(t)-\frac{\partial\Phi(\varepsilon,t)}{\partial t}+\dot{\epsilon}(t)$ is a bounded function.

\begin{theorem}
\label{Th5}
Let conditions (a)-(d) hold for transformation \eqref{eq2_2} and
$\frac{\partial\Phi(\varepsilon,t)}{\partial \varepsilon}>0$ for all $\varepsilon$ and $t$.
Given $\alpha>0$ and $K>0$ there exist $\beta>0$, $\mu_0>0$ and $a$ such that for $\mu<\mu_0$ the polynomial $\lambda(\mu \lambda+1)^{\rho-1}+a\mu$ is Hurwitz and LMI \eqref{LMI_Th02} holds. Then control law \eqref{eq6_5} ensures goal \eqref{eq2_20}.
\end{theorem}

\begin{proof}
Expression \eqref{eq6_60} is the differential equation with regular perturbation, where $\mu$ is the small parameter. According to \cite{Vasilieva73,Bauer15}, let us study \eqref{eq6_60} for $\mu=0$. To this end, rewrite  \eqref{eq6_60} in the form
\begin{align}
\dot{\overline\varepsilon}=\left(\frac{\partial\Phi(\overline\varepsilon,t)}{\partial \overline\varepsilon}\right)^{-1}
\left(-K\overline\varepsilon
+\phi
\right).
\label{eq6_60a}
\end{align}
For the ISS analysis of \eqref{eq6_60a} consider Lyapunov function
$V=0.5\overline\varepsilon^2$.
Verify the relation $\dot{V}+2\alpha V\left(\frac{\partial\Phi(\overline\varepsilon,t)}{\partial \overline\varepsilon}\right)^{-1}-\beta\left(\frac{\partial\Phi(\overline\varepsilon,t)}{\partial \overline\varepsilon}\right)^{-1}\phi^2<0$.
Substituting $V$ and \eqref{eq6_60a} into the last inequality, we get $-(K-\alpha)\overline\varepsilon^2+\overline\varepsilon \phi-\beta \phi^2 < 0$. If LMI \eqref{LMI_Th02} is feasible, then the last inequality holds.
Therefore, the solution of \eqref{eq6_60a} is bounded. Therefore, according to Theorem 2.2 from \cite{Vasilieva73}, \cite{Bauer15}, there exists $\mu_0$ such that for $\mu<\mu_0$ the condition $|\overline{\varepsilon}(t)-\varepsilon(t)|<O(\mu)$ holds, where $\lim_{\mu \to 0}O(\mu)=0$.
As a result, for $\mu<\mu_0$ the solution of \eqref{eq6_60} is also bounded.
Then, due to the boundedness of $f(t)$, $\varepsilon(t)$ and Hurwitz of $Q(\lambda)$ and $\lambda(\mu \lambda+1)^{\rho-1}+a\mu$, the signal $y(t)$ is bounded from \eqref{eq6_6}. Therefore, the control law $u$ is bounded from \eqref{eq6_5}. According to Theorem \ref{Th1}, goal \eqref{eq2_2} holds. Theorem \ref{Th5} is proved.
\end{proof}

\textit{Example 7.} Let the parameters of \eqref{eq6_2} be defined as
\begin{equation}
\label{ex6_1}
\begin{array}{l} 
A=\begin{bmatrix}
0 & 1 & 0\\
0 & 0 & 1\\
-1 & -3 & -3
\end{bmatrix},~~~
B=\begin{bmatrix}
0\\
0\\
1
\end{bmatrix},
\\
D=[1~1~1]^{\rm T},~~~
L=[1~0~0].
\end{array}
\end{equation}
The disturbance $f(t)$ is given by \eqref{eq2_10}. The relative degree is equal to 3.

Let $K=3$, $\mu=0.01$ and $a=0.1$ in \eqref{eq5_2}. According to \eqref{ex6_1} $Q(p)=(p+1)^3$ and $R(p)=1$. Then, control law \eqref{eq6_5} is rewritten as
\begin{equation}
\label{eq6_50}
\begin{array}{l} 
u=-\frac{3(p+1)^3}{p(0.01 p+1)^{2}+10^{-3}}\varepsilon.
\end{array}
\end{equation}
Choose $\Phi$ as in Example 3, where $\Phi(\varepsilon,t)=\frac{\overline{g}(t)e^{\varepsilon}+\underline{g}(t)}{e^{\varepsilon}+1}$ and
\begin{equation*}
\label{eq6_11}
\begin{array}{l} 
\overline{g}(t)=
\begin{cases}
   2\cos(t)+0.2, & 0 \leq t \leq 2\pi,\\
   2.2, & t>2\pi;
\end{cases}
\\
\underline{g}(t)=
\begin{cases}
   2\cos(t)-0.2, & 0 \leq t \leq 2\pi,\\
   1.8, & t>2\pi.
\end{cases}
\end{array}
\end{equation*}
The simulations show that control law \eqref{eq6_50} provides goal \eqref{eq2_10}. Moreover, control law \eqref{eq6_50} is robust w.r.t. parametric uncertainties. Fig.~\ref{Fig44} shows the transients in $y(t)$ for $x(0)=col\{2,1,1\}$ and non-Hurwitz matrix
$A=\begin{bmatrix}
0 & 1 & 0\\
0 & 0 & 1\\
1 & 3 & 3
\end{bmatrix}$.

\begin{figure}[h]
\center{\includegraphics[width=0.6\linewidth]{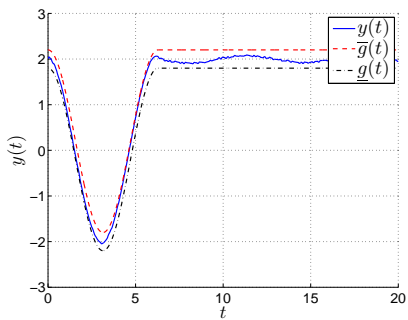}}
\caption{The transients in $y(t)$ for $x(0)=col\{2,1,1\}$.}
\label{Fig44}
\end{figure}

\section{Conclusion}
\label{Sec7}
The method for control of dynamical systems based on a special change of coordinates is proposed. 
According to this method, the initial control problem with the given restriction on an output variable leads to the problem of the input-to-state stability analysis of a new extended system without restrictions. 
As a result, a plant output signal belongs to a given set at any time in the closed-loop system. 
The examples of change of coordinates that can be used for design algorithms are presented. 
Based on the proposed method, the new control laws for linear plants, systems with sector nonlinearity and systems with an arbitrary relative degree are designed. 

The simulations confirm theoretical results. 
The proposed control laws illustrate the effectiveness of the proposed method in the presence of parametric uncertainty and external disturbances. 
Since the plant initial conditions must belong to a given restrictions, in examples the functions specifying restrictions at an initial time are proposed. 
Also, the simulations show that the control law performance can be improved if the design parameters of the loop using a new variable more than the design parameters of the loop using the output signal.

\end{document}